\documentclass[11pt, reqno]{amsart}
\usepackage{graphics}
\usepackage{amsfonts}

\def\Box{\vcenter{\vbox{\hrule\hbox{\vrule
     \vbox to 8.8pt{\hbox to 10pt{}\vfill}\vrule}\hrule}}}

\newcommand{\al}{{\alpha}}
\newcommand{\eps}{{\epsilon}}
\newcommand{\ra}{{\rightarrow}}
\newcommand{\F}{{\mathbb F}}

\newcommand{\Q}{{\mathbb Q}}

\newcommand\C{{\mathcal C}}

\def\Tr{\operatorname{Tr}}

\newtheorem{thm}{Theorem}[section]
\newtheorem{lemma}[thm]{Lemma}

\newtheorem{conj}[thm]{Conjecture}

\def\Tr{\operatorname{Tr}}

\numberwithin{equation}{section}

\begin{document}

\title[Weight Distributions of Binary Cyclic Codes]
{Binary Cyclic codes with two primitive nonzeros}

\author[Tao Feng, Ka Hin Leung, Qing Xiang]
{Tao Feng$^1$, Ka Hin Leung, Qing Xiang}
\thanks{$^1$Research supported in part by Fundamental Research Fund for the Central Universities of China, Zhejiang Provincial Natural Science Foundation (LQ12A01019), National Natural Science Foundation of China (11201418) and Research Fund for Doctoral Programs from the Ministry of Education of China (20120101120089).}

\address{Department of Mathematical Sciences, Zhejiang University, Hangzhou, Zhejiang 310027, P. R. China, email:{\tt tfeng@zju.edu.cn}}

\address{Department of Mathematical Sciences, National University of Singapore, Kent Ridge, Singapore 119260, Republic of Singapore, email:{\tt matkhl@nus.edu.sg}}

\address{Department of Mathematical Sciences, University of Delaware, Newark, DE 19716, USA,
email: {\tt xiang@math.udel.edu}}


\date{}

\begin{abstract}
In this paper, we make some progress towards a well-known conjecture on the minimum weights of binary cyclic codes with two primitive nonzeros.  We also determine the Walsh spectrum of $\Tr(x^d)$ over $\F_{2^{m}}$ in the case where $m=2t$, $d=3+2^{t+1}$ and $\gcd(d, 2^{m}-1)=1$.
\end{abstract}
\maketitle

\section{Introduction}\label{intro}

In this paper, we are concerned with the weight distributions of binary cyclic codes with two primitive nonzeros. Let $q=2^m$,  where $m\geq 1$ is an integer, and $\F=\F_q$, the finite field of size $q$. Let $\al$ be a primitive element of $\F$, and let $\C_{d}$ be the binary cyclic code of length $q-1$ with two nonzeros $\al^{-1}$ and $\al^{-d}$, where $d$ is an integer such that $1\leq d\leq q-2$, $\gcd(d,q-1)=1$. Then $\C_d$ is a $[q-1, 2m]_2$ code, and its codewords are given by
\[
c(a,b)=\big(\Tr(a+b),\Tr(a\al^d+b\al),\ldots,\Tr(a\al^{(q-2)d}+b\al^{q-2})\big),\;a,b\in\F,
\]
where $\Tr$ is the absolute trace function defined on $\F$. 

Let us consider the Hamming weights of $c(a,b)$, where $a,b\in \F$. When exactly one of $a, b$ is $0$, the codeword $c(a,b)$ has weight $q/2$. When $a,b$ are both nonzero, $c(a,b)$ has weight
\begin{align}\label{relation}
 \frac{1}{2}\sum_{i=0}^{q-2}\big(1-(-1)^{\Tr(a\al^{di}+b\al^i)}\big)=
\frac{1}{2}\big(q-\sum_{x\in\F}(-1)^{\Tr(x^d+ba^{-\frac{1}{d}}x)}\big),
\end{align}
where we use $1/d$ to denote the unique integer $j$ such that $jd\equiv 1\pmod{q-1}$ and $1\leq j\leq q-2$. 
Therefore, the weight distribution of $\C_{d} $ is completely determined by the Walsh spectrum of the function $f_d:\F\ra \F_2$, $x\mapsto \Tr(x^d)$, and vice versa. Here {\it the Walsh coefficients} of $f_d$ are defined by
\[
W_d(a)=\sum_{x\in \F}(-1)^{\Tr(x^d+ax)},\quad a\in \F.
\]
The distribution of $W_d(a), a\in \F$, is called {\it the Walsh spectrum} of $f_d$. The problem of determining the Walsh spectrum of $f_d$ is also equivalent to the problem of determining the crosscorrelations of an m-sequence and its $d$-decimation. We refer the reader to the appendix in \cite{DK} for more details on various formulations of this problem. A lot of work has gone into determining the Walsh spectrum of $f_d$ when $d$ takes special forms, see \cite{Niho}, \cite{CD}, \cite{CCD}, and \cite{HX}. There are a few general conjectures on the Walsh spectrum of $f_d$ which have proved to be quite challenging. We refer the reader to the recent paper \cite{AL} for a list of these conjectures, and some recent progress made on them.

In this paper, we are primarily interested in the following well-known conjecture due to Dilip V. Sarwate, cf. \cite{AL}; see \cite[p.~258]{PC} also.

\begin{conj}\label{conj1}Let $m=2t$, and let $\C_d$ be the $[2^m-1,2m]$ binary cyclic code with two nonzeros $\al^{-1}$ and $\al^{-d}$ ($gcd(d,2^m-1)=1$), where $\al$ is a primitive element of $\F$. Then the minimum distance of $\C_d$ is $\leq 2^{m-1}-2^t$. \end{conj}

Using (\ref{relation}), the existence of a nonzero codeword of weight $\leq 2^{m-1}-2^t$ is equivalent to the existence of nonzero $a\in\F$ such that $W_d(a)\geq 2^{t+1}$. Charpin \cite{PC} showed that Conjecture 1.1 is true when $d\equiv 2^j\pmod{2^t-1}$, for some $j$, $0\leq j\leq t-1$.  (Such $d$'s are called the Niho exponents.) \\

In this paper, without putting any conditions on $d$ (of course, $\gcd(d, 2^m-1)=1$ is still assumed), we shall prove an upper bound on the minimum distance of $\C_d$, which is slightly weaker than the bound in Conjecture 1.1. Furthermore, we will determine the weight distributions of $\C_d$ for two special classes of $d$; one of the two classes was previously considered by Cusick and Dobbertin \cite{CD}, the other class is new. Details are given in Section 3. Throughout the rest of this paper, we shall fix $m=2t$. We use $\Tr_m$, $\Tr_t$ to denote the absolute traces defined on $\F$ and $L:=\F_{2^t}$,  respectively. Also we use  $\Tr_{m/t}$ (resp. ${\rm N}_{m/t}$) to denote the relative trace (resp. norm) from $\F$ to $\F_{2^t}$. We shall drop the subscripts if we believe that no confusion will arise.

\section{An upper bound on the minimum weight of $\C_d$}

First, we give a summary of some well-known identities involving the Walsh coefficients $W_d(a)$, $a\in \F$. We refer the reader to \cite{He76, PC,DK,tf} for the proof of these identities.\\

\begin{lemma}\label{known}
(1) $\sum_{a\in\F}W_d(a)=q$, $\sum_{a\in\F}W_d(a)^2=q^2$.\\

(2) $$\sum_{a\in\F_{2^t}}W_d(au)=\begin{cases}q,\text{ if }u\in\F_{2^t}^*;\\0,\text{ if }u\notin \F_{2^t}.\end{cases}$$
\end{lemma}

Now we are ready to prove our first result.

\begin{thm}\label{bound}
Let $m=2t$, and let $\C_d$ be the $[2^m-1,2m]$ binary cyclic code with two nonzeros $\al^{-1}$ and $\al^{-d}$ ($gcd(d,2^m-1)=1$), where $\al$ is a primitive element of $\F$. Then the minimum distance of $\C_d$ is $<2^{m-1}-2^{t-1}-2^{\lfloor t/2\rfloor-1}$. That is, there is a nonzero $a\in \F$ such that $W_d(a)> 2^t+2^{\lfloor t/2\rfloor}$.
\end{thm}

\begin{proof}
 For any nonzero $b\in \F\setminus\F_{2^t}$, by direct calculations we have
\begin{eqnarray}\label{tots}
\sum_{a\in\F_{2^t}}W_d(a)\left(1-(-1)^{\Tr_m(ba)}\eps_b\right)=2^m+2^t|M_b|,
\end{eqnarray}
where $M_b=\sum_{x\in\F_{2^t}}(-1)^{\Tr_m((x+b)^d)}$ and $\epsilon_b=\pm 1$ is chosen such that $\epsilon_b M_b=-|M_b|$. For $b\in \F\setminus\F_{2^t}$, it will be convenient to introduce a function $p_b$ on $\F_{2^t}$ defined by $$p_b(a):=1-(-1)^{\Tr_m(ba)}\eps_b,\;\forall  a\in \F_{2^t}.$$ 
Then for $b\in \F\setminus\F_{2^t}$, we have $\sum_{a\in\F_{2^t}}p_b(a)=2^t$, $p_b(a)\geq 0$, and (\ref{tots}) can be rewritten as
\begin{eqnarray}\label{totsnew}
\sum_{a\in\F_{2^t}}W_d(a)p_b(a)=2^m+2^t|M_b|.
\end{eqnarray}

Next we compute
\begin{align*}\sum_{b\in\F}M_b^2&=2^t\sum_{b\in\F}\sum_{x\in \F_{2^t}}(-1)^{\Tr_m((x+b)^d+b^d)}\\
&=2^t|\F|+2^t\sum_{b\in\F}\sum_{x\in \F_{2^t}^*}(-1)^{\Tr_m(x^d((1+b)^d+b^d))}\\
&=2^t|\F|+2^t\left(2^t\cdot |\{b\in\F\mid \Tr_{m/t}((1+b)^d+b^d)=0\}|-|\F|\right)\\
&=2^{2t}|\{b\in\F\mid (1+b)^d+b^d\in\F_{2^t}\}|.
\end{align*}

Since $M_b=2^t$ if $b\in\F_{2^t}$, we thus have
\[
\sum_{b\in\F\setminus\F_{2^t}}M_b^2=2^{2t}\cdot |\{b\in\F\setminus\F_{2^t}\mid (1+b)^d+b^d\in\F_{2^t}\}|.
\]

Let $c\in\F^*$ be an element of order $2^t+1$. Then a system of coset representatives of $(\F_{2^t}, +)$ in $(\F, +)$ is given by $uc$, $u\in\F_{2^t}$. Since $M_{b+x}=M_b$ for any $x\in\F_{2^t}$, and $\F\setminus\F_{2^t}=\cup_{u\in \F_{2^t}^*} (uc + \F_{2^t})$, we get 
\begin{eqnarray}\label{sumoveru}
\sum_{u\in\F_{2^t}^*}M_{uc}^2=2^t\cdot |\{b\in\F\setminus\F_{2^t}\mid (1+b)^d+b^d\in\F_{2^t}\}|.
\end{eqnarray}

If $u\in\F_{2^t}^*$, we have
\begin{align*}
M_{uc}&=\sum_{x\in\F_{2^t}}(-1)^{\Tr_m((x+uc)^d)}\\
&=\sum_{x\in\F_{2^t}}(-1)^{\Tr_t\left(u^d\left((x+c)^d+(x+c^{2^t})^d\right)\right)}\\
&=\sum_{z\in R_d}\psi_{u^d}(z),
\end{align*}
where $R_d$ denotes the multiset ``$(x+c)^d+(x+c^{2^t})^d$, $x\in\F_{2^t}$" (each element of $R_d$ indeed belongs to $\F_{2^t}$), and $\psi_{u^d}$ is the additive character of $\F_{2^t}$ defined by
\[
\psi_{u^d}(x)=(-1)^{\Tr_t(u^dx)},\;x\in\F_{2^t}.
\]
We write the multiset $R_d$ as a group ring element: $R_d=\sum_{g\in\F_{2^t}}a_g[g]\in\Q[(\F_{2^t},+)]$. Then $\sum_{a\in \F_{2^t}}a_g=2^t$, each $a_g$ is a nonnegative integer, and for $u\in\F_{2^t}^*$, $M_{uc}=\psi_{u^d}(R_d)$. Furthermore note that each coefficient $a_g$ of $R_d$ must be even since $(x+c)^d+(x+c^{2^t})^d=((x+c+c^{2^t})+c)^d+((x+c+c^{2^t})+c^{2^t})^d$ for any $x\in \F_{2^t}$, and $c+c^{2^t}\neq 0$. We compute the coefficient of the identity (i.e., the zero element of $\F_{2^t}$) in $R_dR_d^{(-1)}$ in two ways, where $R_d^{(-1)}=\sum_{g\in \F_{2^t}}a_g[-g]$. In fact, we have $R_d^{(-1)}=R_d$ here.  On the one hand, this coefficient is equal to 
\[
\sum_{g\in\F_{2^t}}a_g^2\geq 2^2\cdot 2^{t-1}=2^{t+1}.
\]
On the other hand, by the inversion formula (see, for example \cite{tf}), the coefficient of the identity element in $R_dR_d^{(-1)}$ is equal to $\frac{1}{2^t}\sum_{u\in\F_{2^t}}\psi_{u^d}(R_d)^2=\frac{1}{2^t}\sum_{u\in\F_{2^t}}M_{uc}^2$. It follows that 
$$\sum_{u\in\F_{2^t}}M_{uc}^2\geq 2^{2t+1}.$$
Using (\ref{sumoveru}) we now obtain
\[
(2^t)^2+2^t\cdot |\{b\in\F\setminus\F_{2^t}\mid (1+b)^d+b^d\in\F_{2^t}\}|\geq 2^{2t+1}.
\]
Therefore
\[
|\{b\in\F\setminus\F_{2^t}\mid (1+b)^d+b^d\in\F_{2^t}\}|\geq 2^t,
\]
with equality if and only if $R_d$ has size $2^{t-1}$ as a set. As a consequence, there exists an element $u\in\F_{2^t}^*$ such that
\[
|M_{uc}|\geq \sqrt{2^{2t}/(2^t-1)}>2^{\lfloor t/2\rfloor}.
\]
Using the above element $uc$ as $b$ in Eqn. (\ref{totsnew}), we see that there is some $a\in\F_{2^t}$ such that $W_d(a)> 2^t+2^{\lfloor t/2\rfloor}$ by an averaging argument. The proof of the theorem is now complete.
\end{proof}

\noindent {\bf Remarks.} (1) In the case where $d=1+2^i$,  for $x\in\F_{2^t}$, we have $\Tr_m((x+b)^d)=\Tr_t(xv)+\Tr_m(b^d)$, where $v=\Tr_{m/t}(b)^{2^i}+\Tr_{m/t}(b)^{2^{-i}}$. Choosing $b\in\F\setminus\F_{2^t}$ such that $\Tr_{m/t}(b)=1$, we have $v=0$, and $|M_b|=2^t$.  We see that Conjecture 1.1 is true in this case by using (\ref{totsnew}).

(2) If $d$ is a Niho exponent, then from \cite[p.~253]{PC} we know that $2^t|W_d(a)$ for all $a\in\F$. Combining this divisibility result with the conclusion of Theorem~\ref{bound} that there is some $a\in\F$ with $W_d(a)> 2^t+2^{\lfloor t/2\rfloor}$, we immediately get $W_d(a)\geq 2^{t+1}$.  The same argument shows that more generally, for any $d$, $1\leq d\leq q-2$, $\gcd(d, q-1)=1$, such that $2^t|W_d(a)$ for all $a\in\F$, Conjecture 1.1 is also true.

\section{The Walsh spectrum of $\Tr(x^d)$ with $d=1+2^i+2^{i+t}$}

In this section, we assume that $d=1+2^i+2^{i+t}$ for some $i$, $0<i<t-1$, and $\gcd(d,2^m-1)=1$. Such a $d$ is not a Niho exponent. First, we show that for any $d$ of the aforementioned form, Conjecture 1.1 is true. Secondly, specializing to the $i=1$ case, i.e., $d=3+2^{t+1}$, we determine the Walsh spectrum of $\Tr(x^d)$ completely.  

For a nonzero integer $n$, we use $v_2(n)$ to denote the highest power of $2$ dividing $n$.
 
\begin{lemma}\label{powerof2}
Let $m=2t$ and $d=1+2^i+2^{i+t}$ for some $i$, $0<i<t-1$, with $\gcd(d, 2^m-1)=1$. Then $v_2(i+1)\geq v_2(t)$.
\end{lemma}
\begin{proof} Since $\gcd(d,2^m-1)=1$, we have $\gcd(2^{i+1}+1,2^t-1)=1$. It follows that $\gcd(2^{i+1}-1,2^t-1)=\gcd(2^{2(i+1)}-1,2^t-1)$. Therefore $\gcd(i+1,t)=\gcd(2(i+1),t)$, which is easily seen to be equivalent to $v_2(i+1)\geq v_2(t)$. The proof is complete.
\end{proof}
 
Let $c$ be a fixed element of $\F^*$ such that $c\neq 1$ and $c^{2^t+1}=1$. Then each element of $\F$ can be written uniquely as $x+yc$ with $x,y\in L:=\F_{2^t}.$ We shall write $\bar{c}:=c^{2^t}$, $\theta:=c+\bar{c}$. Now we compute $W_d(a+b\bar{c})$, where $a,b\in L$. For $x,y\in L$, we have
\begin{align*}
\Tr((x+yc)^d+(a+b\bar{c})(x+yc))&=\Tr(x{\rm N}_{m/t}(x+yc)^{2^i}+y{\rm N}_{m/t}(x+yc)^{2^i}c+ax+by+ayc+bx\bar{c})\\
&=\Tr_t(y(x^2+xy\theta+y^2)^{2^i}\theta)+\Tr_t(ay\theta+bx\theta)\\
&=\Tr_t(yx^{2^{i+1}}\theta+y^{1+2^i}\theta^{1+2^i}x^{2^i})+\Tr_t(y^{1+2^{i+1}}\theta+ay\theta+bx\theta)\\
&=\Tr_t\left((y^{2^{t-i-1}}\theta^{2^{t-i-1}}+y^{1+2^{t-i}}\theta^{1+2^{t-i}}+b\theta)x\right)+\Tr_t(y^{1+2^{i+1}}\theta+ay\theta).
\end{align*}
Therefore,
\begin{align*}
W_d(a+b\bar{c})&=\sum_{y\in L}\sum_{x\in L}(-1)^{\Tr_t\left((y^{2^{t-i-1}}\theta^{2^{t-i-1}}+y^{1+2^{t-i}}\theta^{1+2^{t-i}}+b\theta)x\right)+\Tr_t(y^{1+2^{i+1}}\theta+ay\theta)}\\
&=2^t\sum_{y}(-1)^{\Tr_t(y^{1+2^{i+1}}\theta+ay\theta)},
\end{align*}
where the last sum is taken over
\[\{y\in L\,|\,y\theta+(y\theta)^{2+2^{i+1}}+(b\theta)^{2^{i+1}}=0\}.\]
After a change of variable, we have
\begin{align}\label{wspectrum}
W_d(a+b\bar{c})&=2^t\sum_{z\in S_b}(-1)^{\Tr_t(z^{1+2^{i+1}}\theta^{-2^{i+1}}+az)},
\end{align}
where
\[
S_b:=\{z\in L\,|\,z+z^{2+2^{i+1}}+(b\theta)^{2^{i+1}}=0\}.
\]

When $b=0$, we have $S_0=\{0,1\}$ since $\gcd(2^{i+1}+1, 2^t-1)=1$. It follows that 
\[
W_d(a)=2^t(1+(-1)^{\Tr_t(\theta^{-1}+a)}),\quad \forall a\in L.
\]
Choosing $a=\theta^{-1}$,  we have $W_d(\theta^{-1})=2^{t+1}$. Thus we have proved the following:

\begin{thm}Conjecture \ref{conj1} holds when $d$ is of the form $1+2^i+2^{i+t}$, $0<i<t-1$, and $\gcd(d, 2^m-1)=1$.
\end{thm}

In the case where $b\neq 0$, we need to solve the equation
\[
z+z^{2^{i+1}+2}=w, \quad z\in L,
\]
for each $w\in L^*$. For general $i$, $0<i<t-1$, the solutions are complicated. We will consider the $i=1$ case below.

From now on, we assume that $i=1$ (so $d=3+2^{t+1}$). By Lemma~\ref{powerof2}, $v_2(t)\leq 1$; that is, either $t$ is odd or $t\equiv 2\pmod 4$. The equation we need to consder is now $z^6+z=w$, $z\in L$ and $w\in L^*$.


Assume that $z_0\in L^*$ is a solution to $z^6+z=w$, $w\in L^*$. Suppose $z_0+x$ is another solution with $x\in L^*$. Now expanding $(z_0+x)^6+z_0+x=w$ gives
\[
\left(\frac{x}{z_0}\right)^5+\left(\frac{x}{z_0}\right)^3+\left(\frac{x}{z_0}\right)=\frac{1}{z_0^5}.
\]
The polynomial $X^5+X^3+X\in \F_2[X]$ is the Dickson polynomial $D_5(X,1)$. For convenience of the reader, we include the definition of general Dickson polynomials here. Let $a\in \F_q$ (here $q$ is an arbitrary prime power) and let $n$ be a positive integer. We define the {\it
Dickson polynomial} $D_n(X,a)$ over $\F_q$ by
$$D_n(X,a)=\sum_{j=0}^{\lfloor n/2\rfloor}\frac {n}{n-j}{n-j\choose j}(-a)^jX^{n-2j}.$$
It is well known \cite{dk} that the Dickson polynomial $D_n(X,a)$, $a\in \F_q^*$, is a permutation polynomial of $\F_q$ if and only if $\gcd(n,q^2-1)=1$. For more details about Dickson polynomials, we refer the reader to \cite{dk}.

We are now ready to determine the Walsh spectrum of $\Tr(x^d)$ in the case where $m=2t$, $t$ is odd, and $d=3+2^{t+1}$.

\begin{thm}\label{todd} 
Let $m=2t$ be a positive integer with $t$ odd, and $d=3+2^{t+1}$. The Walsh spectrum of $\Tr(x^{d})$ over $\F=\F_{2^m}$ is given as follows.
\begin{table*}[ht]
\begin{center}
\caption{}\label{tab-comp}
\begin{tabular}{c| c}\hline
$W_d(\cdot)$ & multiplicity\\\hline
$0$ & $3\cdot 2^{2t-2}$\\
$2^{t+1}$& $2^{2t-3}+2^{t-2}$\\
$-2^{t+1}$& $2^{2t-3}-2^{t-2}$\\\hline
\end{tabular}
\end{center}
\end{table*}
\end{thm}
\begin{proof} 
We have observed that $X^5+X^3+X$ is the Dickson polynomial $D_5(X,1)$. If $t$ is odd, then $\gcd(5, 2^{2t}-1)=1$; consequently $D_5(X,1)$ induces a permutation over $L=\F_{2^t}$. Hence by the computations that we did above, $|S_b|=0$ or $2$ when $t$ is odd and $b\neq 0$. We already saw that $S_0=\{0,1\}$. It follows that $W_d(a+b\bar{c})$, $a,b\in L$, take three values only: $0$, $\pm 2^{t+1}$.   Now denote by $N_0,N_+,N_-$ the multiplicity of  $0$, $2^{t+1}$, $-2^{t+1}$ in the Walsh spectrum of $\Tr(x^d)$, respectively. From part (1) of Lemma \ref{known}, we have
\begin{align*}
N_0+N_++N_-=2^{2t},\quad 2^{t+1}N_+-2^{t+1}N_-=2^{2t},\quad 2^{2t+2}N_++2^{2t+2}N_-=2^{4t}.
\end{align*}
Solving this system of equations, we get
\[
N_0=2^{2t}-2^{2t-2},\quad N_+=2^{2t-3}+2^{t-2},\quad N_-=2^{2t-3}-2^{t-2}.
\]
\end{proof}

\noindent{\bf Remarks.} (1). Let $t$ be an odd positive integer. The fact that $z^6+z=w$, $w\in \F_{2^t}$, has $0$ or $2$ solutions in $L$ is equivalent to the fact that $D(6)=\{(1,x,x^6)\mid x\in \F_{2^t}\}\cup \{(0,1,0), (0,0,1)\}$ is a hyperoval in $PG(2,2^t)$. See \cite{ehkx} for more details.

(2). Theorem~\ref{todd} was first proved in \cite{CD} by a slightly different argument.
\vspace{0.1in} 

Next we consider the case where $d=3+2^{t+1}$ and $t\equiv 2\pmod 4$.
 
\begin{thm}\label{even} Let $m=2t$ be a positive integer with $v_2(t)=1$, $t\geq 6$, and $d=3+2^{t+1}$.  The Walsh spectrum of $\Tr(x^{d})$ over $\F=\F_{2^m}$ is given as follows.
\begin{table*}[ht]
\begin{center}
\caption{}\label{tab-comp}
\begin{tabular}{c| c}\hline
$W_d(\cdot)$ & multiplicity\\\hline
$0$ & $2^{2t-1}-2^{2t-5}-2^{t-1}+2^{t-3}$\\
$2^t$ & $\frac{2^{2t}+2^t}{5}$\\
 $-2^t$ & $\frac{2^{2t}+2^t}{5}$ \\
$2^{t+1} $& $ 2^{2t-4}+2^{t-2}$\\
$-2^{t+1} $& $ 2^{2t-4}-2^{t-2}$\\
$2^{t+2}$ & $\frac{2^{2t-6}-2^{t-4}}{5}$\\
$-2^{t+2}$ &$\frac{2^{2t-6}-2^{t-4}}{5}$\\\hline
\end{tabular}
\end{center}
\end{table*}
\end{thm}
\noindent
The remaining part of this paper is devoted to the proof of Theorem \ref{even}. From now on we always assume that $v_2(t)=1$ and $t\geq 6$. Let $G:=\{x\in \F\mid x^{2^t+1}=1\}$. Furthermore we will assume that the element $c$ used in (\ref{wspectrum}) to have order $5$. Since $t\equiv 2 \pmod 4$ by assumption, we have $5|(2^t+1)$. Thus $c^{2^t+1}=1$, i.e., $c\in G$ (and $c\not\in L$).
 
\begin{lemma}\label{size}
Let $w\in L^*$. Then the number of solutions $z\in L$ to 
$$z^6+z=w$$
is $0$, $1$, $2$ or $6$.
\end{lemma}
\begin{proof} The main difference from the $t$ odd case is that $X^5+X^3+X\in \F_2[X]$ no longer induces a permutation of $L=\F_{2^t}$ when $t\equiv 2\pmod 4$. We start in the same way as before. Assume that $z_0\in L^*$ is a solution to $z^6+z=w$, $w\in L^*$. Suppose $z_0+x$ is another solution with $x\in L^*$. Then expanding $(z_0+x)^6+z_0+x=w$ gives 
\begin{eqnarray}\label{degree5}
\left(\frac{x}{z_0}\right)^5+\left(\frac{x}{z_0}\right)^3+\left(\frac{x}{z_0}\right)=\frac{1}{z_0^5},
\end{eqnarray}
which has 0, $1$, or $5$ solutions in $L$ when $v_2(t)=1$ and $t\geq 6$. This can be seen as follows.

It is well known that each element $y$ of $L^*$  can be written in the form $u+\frac{1}{u}$, with $u\in L^*$ or $u\in G$, according as $\Tr_t(1/y)$ is equal to $0$ or $1$ (see \cite{dk}). Now if $ u+\frac{1}{u}\in L$ is a solution to (\ref{degree5}), then so are $ \gamma u+\frac{1}{\gamma u}$, $\gamma\in \F^*$ and $\gamma^5=1$, since $D_5(u+\frac{1}{u},1)=u^5+\frac{1}{u^5}$. When $u\in L^*$, $\gamma u+\frac{1}{\gamma u}$ is in $L$ if and only if $\gamma=1$.  When $u\in G$,  any choice of $\gamma$ ($\gamma^5=1$) will give $\gamma x+\frac{1}{\gamma x}\in L$. This proves the claim that (\ref{degree5}) has $0$, $1$ or $5$ solutions in $L$. The conclusion of the lemma follows as a consequence.
\end{proof}

From Lemma~\ref{size} and (\ref{wspectrum}), we see that the Walsh coefficients of $\Tr(x^{3+2^{t+1}})$ are in $\{\pm i\cdot 2^{t}\mid i=0, 1, 2, 4, 6\}$. We use $N_i$ to denote the number of $a+b\bar{c}\in \F$ such that $W_d(a+b\bar{c})=i\cdot 2^{t}$, for $i\in\{0,\pm1,\pm 2,\pm 4,\pm 6\}$.

\subsection{} 
 Now, we examine for which $w\in L^*$, $z^6+z=w$, has six solutions in $L$. Assume that $z_0$ and $x$ are as in the proof of Lemma \ref{size}. By the above analysis, there exists $u\in G$ such that $\frac{x}{z_0}=u+\frac{1}{u}$, and $\frac{1}{z_0^5}=u^5+\frac{1}{u^5}$, i.e., $z_0^5=\frac{1}{u^{-5}+u^5}$. Since $\gcd(5,2^t-1)=1$, we get $z_0=\frac{1}{(u^{-5}+u^5)^{1/5}}$. The other five solutions are
\[
\frac{1}{(u^{-5}+u^5)^{1/5}}\left(1+u\gamma+\frac{1}{u\gamma}\right),\quad\gamma^5=1.
\]
Therefore, $z^6+z=w$, $w\in L^*$,  has six solutions in $L$ if and only if $w$ is in the following set
\[
T_6:=\{z^6+z\,|\,z=\frac{1}{(u^{-5}+u^5)^{1/5}},\, u\in G,\,u^5\neq 1\}.
\]
The set $T_6$ has size $\frac{2^t+1-5}{5\cdot 2\cdot 6}=\frac{2^{t-2}-1}{15}$ : the factor $5$ in the denominator comes from the fact that $u\mapsto u^5$ is $5$-to-$1$ on $G$;  the factor $6$  comes from the fact $z\mapsto z^6+z$ is $6$-to-$1$ on the set in consideration; and the factor $2$ comes from the fact $u$ and $u^{-1}$ give the same element. In this case, with $(b\theta)^4=w$, $W_d(a+b\bar{c})\in\{\pm i\cdot 2^{t}\mid  i=0, 2, 4, 6\}$. \\

Next, we examine for which $w\in L$, $z^6+z=w$ has two solutions in $L$. Clearly, when $w=0$, this equation has two solutions in $L$. So in what follows we consider the case where $w\neq 0$. Assume that $z_0$ and $x$ are as in the proof of Lemma \ref{size}. By the same analysis, there exists $u\in L^*$ such that $\frac{x}{z_0}=u+\frac{1}{u}$, and $\frac{1}{z_0^5}=u^5+\frac{1}{u^5}$, i.e., $z_0^5=\frac{1}{u^{-5}+u^5}$.  Therefore, $z^6+z=w$, $w\in L$,  has two solutions in $L$ if and only if $w$ is in the following set
\[
T_2:=\{z^6+z\,|\,z=\frac{1}{(u^{-5}+u^5)^{1/5}},\, u\in L\setminus\F_4\}\cup\{0\}.
\]
The set $T_2$ has size $\frac{2^t-4}{2\cdot 2}+1=2^{t-2}$. In this case,  with $(b\theta)^4=w$, $W_d(a+b\bar{c})\in\{\pm i\cdot 2^{t}: i=0, 2\}$.\\

It now follows that there are $2^t-2\cdot 2^{t-2}-6\cdot \frac{2^t-4}{60}=\frac{2^{t+1}+2}{5}$ elements $w\in L$ such that $z^6+z=w$ has only one solution in $L$. Only these $w$ will give the values $W_d(a+b\bar{c})=\pm 2^{t}$ (again with $(b\theta)^4=w$). We observe that the two values, $2^t$ and $-2^t$, occur for equally many $a\in L$, since for the unique solution $z_0\in L^*$ to $z^6+z=w$, half of the $a$'s in $L$ satisfy $\Tr_t(az_0)=0$ and the other half satisfy $\Tr_t(az_0)=1$. Therefore we have
\[
N_1=N_{-1}=2^{t-1}\cdot\frac{2^{t+1}+2}{5}=\frac{2^{2t}+2^t}{5}.
\]

Finally we note that the number of $w\in L$ such that  $z^6+z=w$ has no solutions in $L$ at all is equal to $2^{t}-\frac{2^{t-2}-1}{15}-2^{t-2}-\frac{2^{t+1}+2}{5}=\frac{2^t-1}{3}$.

\subsection{}
We now show that $W_d(a+b\bar{c})\neq \pm 6\cdot 2^{t}$ for all $a,b\in L$. As seen above, only when $z^6+z=w$, $w=(b\theta)^4\in L^*$, has 6 solutions in $L$, could $W_d(a+b\bar{c})$ possibly be equal to $\pm 6\cdot 2^{t}$. Let $z_0=\frac{1}{(u^{-5}+u^5)^{1/5}}\in L^*$, $u\in G$, be a solution to $z^6+z=w$, $w=(b\theta)^{4}\in L^*$. The other five solutions are $z_j=z_0+x_j\in L$, with $\frac{x_j}{z_0}=u\gamma^j+\frac{1}{u\gamma^j}$, $1\leq j\leq 5$,  $o(\gamma)=5$, $u\in G$. The fact that $\pm 6\cdot 2^{t}$ won't occur as Walsh coefficients of $\Tr(x^d)$ amounts to the fact that the following system of equations does not have a solution $a\in L$:
\[
\Tr_t(z_j^5\theta^{-4}+az_j)=\Tr_t(z_0^5\theta^{-4}+az_0),\quad 1\leq j\leq 5.
\]
We will prove the latter fact by way of contradiction. Assume that the above system has a solution $a\in L$. With $z_j=x_j+z_0$, we get
\[
\Tr_t\left(x_j\left(z_0^4\theta^{-4}+z_0^{2^{t-2}}\theta^{-1}+a\right)\right)=\Tr_t(x_j^5\theta^{-4}),\quad 1\leq j\leq 5.
\]
Since $\frac{x_j}{z_0}=u\gamma^j+\frac{1}{u\gamma^j}=\Tr_{m/t}(u\gamma^j)$, we have
\[
\Tr_{m}\left(u\gamma^j z_0\left(z_0^4\theta^{-4}+z_0^{2^{t-2}}\theta^{-1}+a\right)\right)=\Tr_{m}\left(\left(u^5+u^3\gamma^{3j}\right)z_0^5\theta^{-4}\right),\quad 1\leq j\leq 5.
\]
Now, we rewrite the above equations as
\[
\Tr_{4}\left(\gamma^jU\right)=V+\Tr_{4}\left(\gamma^{3j}W\right),\quad 1\leq j\leq 5.
\]
where
\begin{align*}
U:&=\Tr_{m/4}( uz_0(z_0^4\theta^{-4}+z_0^{2^{t-2}}\theta^{-1}+a))=\Tr_{m/4}\left( \frac{u}{u^5+u^{-5}}\theta^{-4}+\frac{u}{(u^5+u^{-5})^{1/4}}\theta^{-1}+uz_0a\right),\\
V:&=\Tr_{m}(u^5z_0^5\theta^{-4})=\Tr_{m}\left(\frac{u^5}{u^5+u^{-5}}\theta^{-4}\right)=\Tr_t(\theta^{-1}),\\
W:&=\Tr_{m/4}(u^3z_0^5\theta^{-4})=\Tr_{m/4}\left(\frac{u^3}{u^5+u^{-5}}\theta^{-4}\right).
\end{align*}

Taking summation of the above equations over $1\leq j\leq 5$, we get $V=0$. However, as we stated before, $\Tr_t(\theta^{-1})=1$ since $\theta=c+c^{-1}$ with $c\in G$. This contradiction completes the proof. \\

\subsection{}
(1) We now compute $N_4$ and $N_{-4}$. As we have seen above, $W_d(a+b\bar{c})=\pm 2^{t+2}$ if and only if $z^6+z=w$, $w=(b\theta)^4\in L^*$, has 6 solutions in $L$, and for some $i_0\in\{0,1,\ldots,5\}$ the following equations hold:
\[
\Tr_t(z_j^5\theta^{-4}+az_j)=\Tr_t(z_{i_0}^5\theta^{-4}+az_{i_0})+1,\quad 0\leq j\leq 5,\,j\neq i_0.
\]
Without loss of generality we may assume that $i_0=0$. Similar to the above computations, we can rewrite the above equations as
\[
\Tr_{4}\left(\gamma^jU\right)=\Tr_{4}\left(\gamma^{3j}W\right),\quad 1\leq j\leq 5,
\]
where $U,W$ are the same as above. It follows that
\[\Tr_{4}\left(\gamma^jU\right)=\Tr_{4}\left(\gamma^{j}W^2\right),\quad 1\leq j\leq 5.\]
Since $\gamma^j$, $1\leq j\leq 5$, span $\F_{2^4}$, we obtain that $U=W^2$, i.e.,
\begin{align*}
\Tr_{m/4}(uz_0a)
   &=\Tr_{m/4}\left( \frac{u}{(u^5+u^{-5})^{1/5}}a\right)\\
  &=\Tr_{m/4}\left(\frac{u}{u^5+u^{-5}}\theta^{-4}+\frac{u}{(u^5+u^{-5})^{1/4}}\theta^{-1}+\frac{u^6}{u^{10}+u^{-{10}}}\theta^{-8}\right)
\end{align*}
By assumption $c$ has order $5$, it follows that $\theta=c+\bar{c}$ has order $3$. We have
\begin{align*}
\Tr_{m/4}(uz_0a)
 &=\Tr_{m/4}\left(\frac{u}{u^5+u^{-5}}\theta^{2}+\frac{u}{(u^5+u^{-5})^{1/4}}\theta^{2}+\frac{u^6}{u^{10}+u^{-{10}}}(\theta^2+1)\right)\\
  &=\theta^2\Tr_{m/4}\left(\frac{u}{u^5+u^{-5}}+\frac{u^{16}}{u^{20}+u^{-20}}+\frac{u^6}{u^{10}+u^{-10}}\right)+\Tr_{m/4}\left(\frac{u^6}{u^{10}+u^{-10}}\right)\\
  &=\theta^2\Tr_{m/4}\left(\frac{u}{u^5+u^{-5}}+\frac{u^{-4}}{u^{20}+u^{-20}}\right)+\Tr_{m/4}\left(\frac{u^3}{u^{5}+u^{-5}}\right)^2\\
  &=\theta^2\Tr_{m/4}\left(\frac{u+u^{-1}}{u^5+u^{-5}}\right)+\theta^2\Tr_{m/2}\left(\frac{u^{-1}}{u^{5}+u^{-5}}\right)+\Tr_{m/4}\left(\frac{u^3}{u^{5}+u^{-5}}\right)^2\\
  &=\theta^2\Tr_{t/2}\left(\frac{u+u^{-1}}{u^5+u^{-5}}\right)+\theta^2\Tr_{t/2}\left(\frac{u+u^{-1}}{u^{5}+u^{-5}}\right)+\Tr_{m/4}\left(\frac{u^3}{u^{5}+u^{-5}}\right)^2\\
  &=\Tr_{m/4}\left(\frac{u^3}{u^{5}+u^{-5}}\right)^2.
\end{align*}
  
Conversely, if $\Tr_{m/4}(uz_0a)=\Tr_{m/4}\left(\frac{u^3}{u^{5}+u^{-5}}\right)^2$, $a\in L$, and $z^6+z=w$, $w=(b\theta)^4\in L^*$, has 6 solutions in $L$, then $W_d(a+b\bar{c})=\pm 2^{t+2}$.

Below we will count the number of solutions to
\begin{eqnarray}\label{trace}
\Tr_{m/4}(uz_0a)=\Tr_{m/4}\left(\frac{u^3}{u^{5}+u^{-5}}\right)^2,\,a\in L.
\end{eqnarray}
Write $\Tr_{m/4}\left(\frac{u^3}{u^{5}+u^{-5}}\right)^2=h+g\gamma$ with $h,g\in\F_{2^2}$ and
\[
uz_0=\frac{u}{(u^5+u^{-5})^{1/5}}=\alpha+\beta\gamma,\quad \alpha,\beta\in L=\F_{2^t}, \,o(\gamma)=5.
\]
We claim that $\alpha/\beta\not\in\F_4^*$. Otherwise, $u$ is in $\F_{2^4}^*\cdot \F_{2^t}^*$ and thus has order dividing ${\rm lcm}(15,2^t-1)=5(2^t-1)$. Noting that $u$ has order dividing $2^t+1$, we have $u^5=1$, which is a contradiction. Now (\ref{trace}) becomes
$\Tr_{m/4}(\alpha a)+\Tr_{m/4}(\beta a)\gamma=h+g\gamma$, that is,
\[
\Tr_{t/2}(\alpha a)=h,\quad \Tr_{t/2}(\beta a)=g.
\]
Since $\alpha/\beta\not\in \F_4^*$, this system of equations clearly has $2^{t-4}$ solutions $a\in L$.

We thus have
\[
N_4+N_{-4}=6\cdot2^{t-4}\cdot\frac{2^{t-2}-1}{15}=\frac{2^{2t-5}-2^{t-3}}{5}.
\]
\vspace{0.1in}

\noindent (2) Let $b\in L^*$ be such that $z^6+z=w$, $w=(b\theta)^4\in L^*$, has 6 solutions in $L$. Assume that the six solutions are $z_j$, $0\leq j\leq 5$, as given above. We claim that for each $i_0\in \{0,1,\ldots ,5\}$ there exists an $x\in L$ such that
\begin{eqnarray}\label{system}
\Tr_{m/4}\left( uz_{i_0}x\right)=0,\; \Tr_t(z_jx)=1,\quad\forall j, \,0\leq j\leq 5.
\end{eqnarray}
An immediate consequence is that $N_4=N_{-4}$; this can be seen as follows: If $W_d(a+b\bar{c})=4\cdot 2^t$, $a,b\in L$, then $W_d(x+a+b\bar{c})=-4\cdot 2^t$ since every term in the sum on the right hand side of (\ref{wspectrum}) is negated and $\Tr_{m/4}(uz_{i_0}(x+a))=\Tr_{m/4}(uz_{i_0}a)=\Tr_{m/4}\left(\frac{u^3}{u^{5}+u^{-5}}\right)^2$.  We thus conclude that
\[
N_4=N_{-4}=\frac{2^{2t-6}-2^{t-4}}{5}.
\]
Now we prove the claim about the existence of solution of (\ref{system}). Again, without loss of generality we assume that $i_0=0$. Multiplying both sides of $\Tr_{m/4}\left( uz_0x\right)=0$ by  $\gamma^j$ and  taking trace to $\F_2$, we get
\[
\Tr_{t}\left(x_jx\right)=0,\quad\forall\,1\leq j\leq 5.
\]
As above, writing $uz_0=\alpha+\beta\gamma$, $\alpha,\beta\in L$, $o(\gamma)=5$, and noting that $z_j=x_j+z_0$, for $1\leq j\leq 5$, we see that the system of equations under consideration reduces to
\[
\Tr_{t/2}(\alpha x)=0,\quad \Tr_{t/2}(\beta x)=0,\quad \Tr_t(z_0x)=1,
\]
We prove that this system of equations has a solution by showing that $z_0$ does not lie in the $\F_4$-linear span of $\alpha$ and $\beta$. Raising $uz_0=\alpha+\beta\gamma$ to the $2^t$-th power gives $u^{-1}z_0=\alpha+\beta\gamma^{-1}$. We solve that
\[
\alpha=\frac{u\gamma^{-1}+u^{-1}\gamma}{\gamma+\gamma^{-1}}z_0,\quad
\beta=\frac{u+u^{-1}}{\gamma+\gamma^{-1}}z_0.
\]
Suppose to the contrary that there exist $r,s\in\F_4$ such that $r\alpha+s\beta=z_0$. After expansion we get
\[
u^2(r+s\gamma^{-1})+u(\gamma+\gamma^{-1})+(r+s\gamma)=0.
\]
This is a degree $2$ equation with coefficients in $\F_{2^4}$. Since $u\in\F_{2^{2t}}$ and $2||t$, we have $u\in\F_{16}^*$. Hence $u^5=1$, which is impossible.\\

\subsection{}
It remains to determine $N_0$, $N_2$, $N_{-2}$.  By Lemma~\ref{known}, we have the following equations
\begin{align*}
N_0+N_2+N_{-2}&=2^{2t}-\frac{2^{2t-5}-2^{t-3}}{5}-2\cdot\frac{2^{2t}+2^t}{5}=19\cdot 2^{2t-5}-3\cdot2^{t-3};\\
2^{t+1}(N_2-N_{-2})&=2^{2t};\\
2^{2t+2}(N_2+N_{-2})&=2^{4t}-\frac{2^{2t-5}-2^{t-3}}{5}\cdot 2^{2t+4}-2\cdot\frac{2^{2t}+2^t}{5}\cdot 2^{2t}=2^{4t-1}.
\end{align*}
Solving these equations, we get
\[
N_0=2^{2t-1}-2^{2t-5}-2^{t-1}+2^{t-3},\quad N_2=2^{2t-4}+2^{t-2},\quad N_{-2}=2^{2t-4}-2^{t-2}.
\]
The proof of Theorem~\ref{even} is now complete.

\end{document}